\documentclass[11pt,twoside]{article}
\usepackage[margin=1in]{geometry}
\usepackage{cite}
\usepackage{amsmath,amssymb,amsfonts, amsthm, thmtools}
\usepackage{algorithmic}
\usepackage{graphicx}
\usepackage{textcomp}
\usepackage{xcolor}
\usepackage{subfigure}
\usepackage{url,hyperref}
\usepackage{enumitem}
\usepackage{authblk}
\def\BibTeX{{\rm B\kern-.05em{\sc i\kern-.025em b}\kern-.08em
    T\kern-.1667em\lower.7ex\hbox{E}\kern-.125emX}}

\usepackage{algorithm}

\allowdisplaybreaks

\newtheorem{theorem}{Theorem}
\newtheorem{lemma}{Lemma}

\newtheorem{definition}{Definition}

\newtheorem{property}{Property}

\newtheorem{remark}[theorem]{Remark}

\newcommand{\R}{\ensuremath{\mathbb{R}}}
\newcommand{\p}[1]{\bold{#1}}
\newcommand{\eps}{\epsilon}
\newcommand{\cala}{\mathcal{A}}
\newcommand{\calr}{\mathcal{R}}
\newcommand{\cost}{{\rm cost}}

\begin{document}

\title{Reliable Distributed Clustering with Redundant Data Assignment}

\author{Venkata Gandikota$^\dagger$, Arya Mazumdar$^\dagger$, and Ankit Singh Rawat$^\ddagger$\thanks{Arya Mazumdar and Venkata Gandikota's research is supported by NSF awards CCF 1642658, CCF 1618512, CCF 1909046 and CCF 1934846}}
\affil{\normalsize $^\dagger$College of Information and Computer Sciences, University of Massachusetts, Amherst, MA\\ $^\ddagger$Google Research New York, NY\\
  \texttt{gandikota.venkata@gmail.com, arya@cs.umass.edu, ankitsrawat@google.com}.}

\sloppy

\maketitle

\begin{abstract}
In this paper, we present distributed generalized clustering algorithms that can handle large scale data across multiple machines in spite of straggling or unreliable machines. We propose a novel data assignment scheme that enables us to obtain global information about the entire data even when some machines fail to respond with the results of the assigned local computations. The assignment scheme leads to distributed algorithms with good approximation guarantees for a variety of clustering and dimensionality reduction problems.\\
\end{abstract}


\section{Introduction}

Clustering is one of the most basic unsupervised learning tools developed to infer informative patterns in data. Given a set of data points in $\R^d$, the goal in clustering problems is to find a smaller number of points, namely {cluster centers}, that form a good representation of the entire dataset. The quality of the clusters is usually measured using a cost function, which is the sum of the distances of the individual points to their closest cluster center.
 
With the constantly growing size of datasets in various domains, centralized clustering algorithms are no longer desirable and/or feasible, which highlights a need to design efficient distributed algorithms for the clustering task. In a distributed setting, we assume that a collection of data points $P$ is too large to fit in the memory of a single server. Therefore, we employ a setup with $s$ {\em compute nodes} and one {\em coordinator server}. In this setup, the most natural approach is to partition the data point in $P$ into $s$ subsets $\{ P_1,\ldots, P_s \} \subset P$ and assign each of these parts to a different compute node. These nodes then perform local computation on the data points assigned to them and communicate their results to the coordinator. The coordinator then combines all the local computation received from the compute nodes and outputs the final clustering. Here, we note that the overall computation of clustering may potentially involve multiple rounds of communications between the compute nodes and the coordinator. This natural approach for distributed clustering has received significant attention from the research community. Interestingly, it is possible to obtain a clustering in the distributed setup that has its cost bounded by a constant multiple of the cost of clustering achievable by a centralized algorithm (see \cite{GMMO00, ABW17a, MKCWM15} and references therein).

In this paper, we aim to address the issue of stragglers that arises in the context of large scale distributed computing systems, especially the ones running on the so-called cloud. The stragglers correspond to those compute nodes that take significantly more time than expected (or fail) to complete and deliver the computation assigned to them. Various system-related issues lead to this behavior, including unexpected background tasks such as software updates being performed on the compute nodes, power outages, and congested communication networks in the computing setup. Some simple approaches used to handle stragglers include ignoring them and relying on asynchronous methods. The loss of information arising due to the straggler nodes can be traded for efficiency for specific tasks such as computing distributed gradients \cite{LLPPR15, TLDK17, dutta18a}. However, with the existing methods for unsupervised learning tasks such as clustering or dimensionality reduction, ignoring the stragglers can lead to extremely poor quality solutions.

Alternatively, one can distribute data to the compute nodes in a redundant manner such that the information obtained from the non-straggler nodes is sufficient to compute the desired function on the entire dataset. Following this approach, multiple coding based solutions  (mainly focusing on the linear computation and first-order methods for optimization) have been recently proposed (e.g., see \cite{LLPPR15, TLDK17, KarakusSDY17, dutta18a, Yu18a, DuttaCG16}).

This paper focuses on the relatively unexplored area of designing straggler-resilient unsupervised learning methods for distributed computing setups. We study a general class of distributed clustering problems in the presence of stragglers. In particular, we consider the $k$-medians clustering problem (cf.~Section~\ref{subsec:kmeds}) and the \emph{subspace clustering} problem (cf.~Section~\ref{subsec:pca}). Note that the subspace clustering problem covers both the $k$-means and the principal component analysis (PCA) problems as special cases. The proposed Straggler-resilient clustering methods in this paper rely on a redundant data distribution scheme that allows us to compute provably good-quality cluster centers even in the presence of a large number of straggling nodes (cf.~Section~\ref{subsec:assignment} and \ref{subsec:construction}). In Section~\ref{sec:experiments}, we empirically evaluate our proposed solution for $k$-median clustering and demonstrate its utility.


\section{Background}\label{sec:background}
In this section, we first provide the necessary background on the clustering problems studied in the paper. We then formally state the main objective of this work.

\subsection{Distributed clustering}
Given a dataset with $n$ points $P = \{\p{p}_1,\ldots, \p{p}_n\} \subseteq \R^d$, distributed among $s$ compute nodes, the goal in distributed clustering problems is to find a set of $k$ cluster centers $C = \{ \p{c_1}, \ldots, \p{c_k} \} \subset \R^d$ that closely represent the entire dataset. The quality of these centers (and the associated clusters) is usually measured by a cost function ${\rm cost}(P, C)$. Two of the most prevalent cost functions for clustering are the $k$-median and the $k$-means functions, which are defined as follows.
\begin{enumerate}
\item \textbf{$k$-median}: ${\rm cost}(P, C) = \sum_{i \in [n]} d(\p{p}_i, C)$,
\item \textbf{$k$-means}: ${\rm cost}(P, C) = \sum_{i \in [n]} d^2(\p{p}_i, C)$,
\end{enumerate}
where, $d(\p{x}, \p{y})$ denotes the Euclidean distance between two points $\p{x}, \p{y} \in \R^d$ and $d(\p{x}, C) := \min_{\p{c} \in C} d(\p{x}, \p{c})$. 
We denote the cluster associated with the center  $\p{c} \in C$ by 
$ {\rm cluster}(\p{c}, P) := \{ \p{x} \in P : \p{c} = \arg\min_{\p{c'} \in C} d(\p{x}, \p{c'}) \}$. 
For any $\alpha> 1$, the set of cluster centers $C$, is called an $\alpha$-approximate solution to the clustering problem if the cost of clustering $P$ with $C$, ${\rm cost}(P, C)$, is at most $\alpha$ times the optimal (minimum) clustering cost with $k$-centers. 

In certain applications, the dataset $P$ is \emph{weighted} with an associated non-negative weight function $w:P\to\R$. The $k$-means cost for such a weighted dataset $(P, w)$ is defined as ${\rm cost}(P, C, w) = \sum_{i \in [n]} w(\p{p}_i) ~d^2(\p{p}_i, C).$ The $k$-median cost for $(P, w)$  is analogously defined.

We also consider a general class of $\ell_2$-error fitting problems  known as the $(r, k)$-subspace clustering problem. 
\begin{definition}[$(r, k)$-subspace clustering]
Given a dataset $P \subset \R^d$ find a set of $k$-subspaces (linear or affine) ${\cal L} = \{L_i \}_{i=1}^k$, each of dimension $r$, that minimizes
$
{\rm cost}(P, {\cal L}) := \sum_{i=1}^n \min_{L \in {\cal L} }d^2(\p{p}_i,  L).
$
\end{definition}

Note that for $r=0$, this is exactly the $k$-means problem described above. Another special case, when $k=1$, is known as principal component analysis (PCA). If we consider the matrix $M \in \R^{n \times d}$, with the data points in $P$ as its rows, it is well-known that the desired subspace is spanned by the top $r$-right singular vectors of $M$.

\subsection{Coresets and clustering}
In a distributed computing setup, where $i$-th compute node stores a partial dataset $P_i \subset P$, one way to perform distributed clustering is to have each node communicate a summary of its local data to the coordinator. An approximate solution to the clustering problem can then be computed from the combined summary received from all the compute nodes. This summary, called a \emph{coreset}, is essentially a weighted set of points that approximately represents the original set of points in $P$.

\begin{definition}[$\eps$-coreset]
\label{def:coreset}
For $0 < \eps < \frac13$,  an $\eps$-coreset for a dataset $P$ with respect to a cost function ${\rm cost}(\cdot, \cdot )$ is a weighted dataset $S$ with an associated weight function $w: S \rightarrow \R$ such that, for any set of centers $C$, we have 
\begin{align*}
(1-\eps)~{\rm cost}(P, C) &\le {\rm cost}(S, C, w)   \le (1+\eps)~{\rm cost}(P, C).
\end{align*}
\end{definition}

The next results shows the utility of a coreset for clustering. 
\vspace{-1em}
\begin{theorem}[~\cite{feldman13}]
\label{thm:coreset-clustering}
Let $(S, w)$ be an $\eps$-coreset for a dataset $P$ with respect to the cost function ${\rm cost}(\cdot, \cdot)$. Any $\alpha$-approximate solution to the clustering problem on input $S$, is an $\alpha(1+3\eps)$-approximate solution to the clustering problem on $P$.
\end{theorem}

\subsection{Straggler-resilient distributed clustering}
The main objective of this paper is to design the distributed clustering methods that are robust to the presence of straggling nodes. Since the straggling nodes are unable to communicate the information about their local data, the distributed clustering method may miss valuable structures in the dataset resulting from this information loss. This can potentially lead to clustering solutions with poor quality (as verified in Section~\ref{sec:experiments}).

Given the prevalence of the stragglers in modern distributed computing systems, it is natural to desire clustering methods that generate provably good clustering solutions despite the presence of stragglers. Let ${\rm OPT}$ be the cost of the best clustering solution for the underlying dataset. In this paper, we explore the following question: {\em Given a dataset $P$ and distributed computing setup with $s$ compute nodes where at most $t$ nodes may behave as stragglers, can we design a clustering method that generates a solution with the cost at most $c\cdot {\rm OPT}$, for a small approximation factor $c \geq 1$?}

In this paper, we affirmatively answer this question for the $k$-median clustering and the $(r, k)$-subspace clustering. Our proposed solutions add on to the growing literature on straggler mitigation via coded computation, which has primarily focused on the supervised learning tasks so far.


\section{Main Results}
We propose to systematically modify the initial data assignment to the compute nodes in order to mitigate the effect of stragglers. In particular, we employ redundancy in the assignment process and map every vector in  the dataset $P$ to multiple compute nodes. 
This way each vector affects the local computation performed at multiple compute nodes, which allows us to obtain final clusters at the coordinator server by taking into account the contribution of most of the vectors in $P$ even when some of the compute nodes behave as stragglers.

We first introduce the assignment schemes with {\em straggler-resilience property}. This property enables us to combine local computations from non-straggling compute nodes at the coordinator while preserving most of the relevant information present in the dataset $P$ for the underlying clustering task. Subsequently, we utilize such an assignment scheme to obtain good-quality solutions to the $k$-medians and the $(r, k)$-subspace clustering problem in Section~\ref{subsec:kmeds} and Section~\ref{subsec:pca}, respectively. Finally, in Section~\ref{subsec:construction}, we propose a randomized construction of an assignment scheme with the desired straggler-resilience property.

\subsection{Straggler-resilient data assignment}
\label{subsec:assignment}

Let the compute nodes in the system be indexed by the set $[s] := \{1, \ldots, s\}$. Furthermore let $\p{p} \in P$ be assigned to the compute nodes indexed by the set $\cala_{\p{p}} \subset [s]$. We can alternatively represent the overall data assignment $\cala = \{\cala_{\p{p}}\}_{\p{p} \in P}$ by an {\em assignment matrix} $A \in \{0, 1\}^{s \times n}$, where the $n$ columns and the $s$ rows of $A$ are associated with distinct points in $P$ and distinct compute nodes, respectively. In particular, the $j$-th column of $A$, which corresponds to the data point $\p{p}_j$, is an indicator of $\cala_{\p{p}_j}$, i.e., 
$A_{i, j} = 1 \text{ if and only if }  i \in \cala_{\p{p_j}}.$
For any $i \in [s]$, we denote the set of data points allocated to the $i$-th compute node by  $P_i = \{\p{p} \in P : i \in \cala_{\p{p}}\}$.

Let $\calr \subset [s]$ denote the set of non-straggling compute nodes. We assume that $|\calr| \ge s-t$, where $t < s$ denotes an upper bound on the number of stragglers in the system. Let $A_{\calr} \in \{ 0,1\}^{|\calr| \times n}$ denote the submatrix of $A$ with only the rows corresponding to the non-straggling compute nodes (indexed by $\calr$). For any such set of non-stragglers $\calr$, we require that the assignment matrix $A$ satisfies the following property. 
\begin{property}[Straggler-resilience property]
\label{prop:assignment}
Let $\delta > 0$ be a given constant. For every $\calr \subset [s]$ with $|\calr| \geq s - t$, there exists a recovery vector, $\p{b} = (b_1,\ldots, b_{|\calr|}) \in \R^{|\calr|}, b_i\ge 0 \forall i \in [n],$ such that for some $1 \leq a_i \leq 1 + \delta~~\text{for all}~i \in [n]$,
\begin{align}
\label{eq:assum1}
\p{b}^T A_{\calr} = \p{a} \equiv
(a_1, a_2,\ldots, a_n).
\end{align}
\end{property}

\begin{remark}
\label{rem:gradient-coding}
The straggler-resilience property is closely related to the gradient coding introduced in~\cite{TLDK17}. However, two key points distinguish our work from the gradient coding work. First, the recovery vector $\p{b}$ is restricted to have only non-negative coordinates. Second, and more importantly, the utilization of the redundant data assignment in this work (cf.~Lemma~\ref{lem:weightedcoreset}) differs from that of gradient coding in \cite{TLDK17} where gradient coding is used to recover the full-gradient.
\end{remark}

The following result, which is based on the combinatorial characterization for the assignment scheme enforced by Property~\ref{prop:assignment}, enables us to combine the information received from non-stragglers to generate close to optimal clustering solutions.
\begin{lemma}\label{lem:weightedcoreset} 
Let $P \subset \R^d$ be a dataset distributed across $s$ compute nodes using an assignment matrix $A$ that satisfies Property~\ref{prop:assignment}. Let $\calr$ denote the set of non-straggler nodes. For any $\delta >0$, let $\p{b} \in \R^{|\calr|}$ be the recovery vector corresponding to $\calr$. 
Then, for any set of $k$ centers $C \subset \R^d$, and any weight function $w: P \rightarrow \R$, 
$${\rm cost}(P, C, w) \le \sum_{i \in \calr} b_i ~{\rm cost}(P_i, C, w)
\le (1+\delta) {\rm cost}(P, C, w).$$
\end{lemma}
Equipped with Lemma~\ref{lem:weightedcoreset}, we are now in the position to describe our solutions for the straggler-resilient clustering.


\subsection{Straggler-resilient distributed $k$-median}
\label{subsec:kmeds}
We distributed the dataset $P$ among the $s$ compute nodes using an assignment matrix that satisfies Property~\ref{prop:assignment}. Each compute node sends a set of (weighted) $k$-medians centers of their local datasets which when combined at the coordinator gives a summary for the entire dataset. Thus, the weighted $k$-median clustering on this summary at the coordinator gives a good quality clustering solution for the entire dataset $P$. Algorithm~\ref{fig:distributedClustering} provides a detailed description of this approach. 
\setlength\floatsep{1mm}
\setlength\textfloatsep{1mm}
\setlength\intextsep{1mm}
\begin{algorithm}[h!]
\begin{algorithmic}[1]
\STATE Input: A collection of $n$ vectors $P \subset \R^d$.
\STATE Allocate $P$  to $s$ workers according to $A$ with Property~\ref{prop:assignment}.
\STATE For each $i \in [s]$, construct $Y_i$, the $k$-median centers in $P_{j}$. Define function $w_i:Y_i \rightarrow \R$ as $w_i(\p{c}) := | {\rm cluster}(\p{c}, P_i) |$, for every $\p{c} \in Y_i$.
\STATE Collect $\{ Y_i \}_{i \in \calr}$ from the non-straggling nodes.
\STATE Let $Y := \cup_{i \in \calr} Y_i$. Using $\p{b}$ from \eqref{eq:assum1}, define $w:Y\rightarrow \R$ such that\footnotemark $w(\p{c}) = b_i \cdot w_i(\p{c})$ for all $\p{c} \in Y_i$ and $i \in \calr$.
\STATE Return $\widehat{C}$, the $k$-median cluster centers of $Y$. 
\end{algorithmic}
\caption{Straggler-resilient distributed $k$-medians}
\label{fig:distributedClustering}
\end{algorithm}
\footnotetext{If for some $i_1\neq i_2 \cdots \neq i_u$, $\p{c} \in Y_{i_1} \cap Y_{i_2} \cdots \cap Y_{i_u}$, then we define $w(\p{c}) = b_{i_1} \cdot w_{i_1}(\p{c}) + b_{i_2} \cdot w_{i_2}(\p{c}) \cdots + b_{i_u} \cdot w_{i_u}(\p{c})$.}

Before assessing the quality of $\widehat{C}$ on the entire dataset $P$, we show that for any set of $k$ centers $C$, the cost incurred by the weighted dataset $Y$ is close to the cost incurred by $P$. 
\begin{lemma}\label{lem:kmedianscoreset}
For any set of $k$-centers $C \subset \R^d$
\begin{align*}
&\cost(P, C) - \sum_{i \in \calr} b_i \cost(P_i, Y_i)\\
&\le \cost(Y, C, w) \le 2(1+\delta)\cost(P, C). 
\end{align*}
\end{lemma}

The following result quantifies the quality of the clustering solution returned by Algorithm~\ref{fig:distributedClustering} on the entire dataset $P$.

\begin{theorem}\label{thm:coresets}
Let $C^*$ be the optimal set of $k$-median centers for the dataset $P$. Then, 
${\rm cost}(P, \widehat{C}) \le 3(1+\delta) {\rm cost}(P, C^*).$ 
\end{theorem}
\begin{proof}
Using the lower bound from Lemma~\ref{lem:kmedianscoreset} with $C = \widehat{C}$,
\begin{align*}
\label{eq:thm3_1}
\cost(P, \widehat{C}) 
&\le \cost(Y, \widehat{C}, w) + \sum_{i \in R} b_i \cost(P_i, Y_i) \\
&\overset{(i)}{\le}  \cost(Y, C^*, w) + \sum_{i \in R} b_i \cost(P_i, C^*) \\
&\overset{(ii)}{\le} 2(1+\delta)\cost(P, C^*) + (1+\delta)\cost(P, C^*),
\end{align*}
where $(i)$ follows from the fact that $\widehat{C}$ and $Y_i$ are the optimal set of centers for the weighted dataset $(Y, w)$ and the partial dataset $P_i$, respectively. For $(ii)$, we utilize the upper bound in Lemma~\ref{lem:kmedianscoreset} and Lemma~\ref{lem:weightedcoreset} (with $C = C^*$). 
\end{proof}

In Algorithm~\ref{fig:distributedClustering}, each compute node sends clustering solution on its local data using which the coordinator is able to construct a good summary of the entire dataset $P$ despite the presence of the stragglers. This summary is sufficient to generate a good quality $k$-median clustering solution on $P$. In Section~\ref{subsec:pca}, we show that if each compute node sends more information in the form of a coreset of its local data, the accumulated information at the coordinator is sufficient to solve the more general problem of $(r, k)$-subspace clustering in a straggler-resilient manner.


\subsection{Straggler-resilient distributed $(r,k)$-subspace clustering}
\label{subsec:pca}

In this subsection, we utilize our redundant data assignment with straggler-resilient property to combine local coresets received from the non-straggling nodes to obtain a global coreset for the entire dataset, which further enables us to perform distributed $(r,k)$-subspace clustering in a straggler-resilient manner. In particular, we propose two approaches to perform subspace clustering, which rely on the coresets~\cite{har2004coresets} and the {\em relaxed} coresets~\cite{feldman13, balcan14}, respectively.

\subsubsection{Distributed $(r,k)$-subspace clustering using coresets}
\label{subsubsec:subspace}
Here, we propose a distributed $(r, k)$-subspace clustering algorithm that used the existing coreset constructions from the literature in a black-box manner. Each compute node sends a coreset of its partial data which when re-weighted according to Lemma~\ref{lem:weightedcoreset} gives us a coreset for the entire dataset even in presence of stragglers. Given this global coreset, we can then construct a solution to the underlying $(r, k)$-subspace clustering problem at the coordinator (cf.~Theorem~\ref{thm:coreset-clustering}). The complete description of this approach is given in Algorithm~\ref{fig:subspaceClustering}. 

\vspace{0.1em}
\begin{algorithm}[h!]
\begin{algorithmic}[1]
\STATE Input:  A collection of $n$ vectors $P \subset \R^d$.
\STATE Allocate $P$ to $s$ workers according to $A$ with Property~\ref{prop:assignment}.
\STATE For each $i \in [s]$, find a $\delta$-coreset $(S_i, w_i)$ for $P_{i}$. 
\STATE Collect $\{S_i\}_{i \in \calr}$ at the coordinator.
\STATE Let $S = \cup_{i \in \calr} S_i$. For every $i \in \calr$, scale the weights of the coreset points received from $i$-th node by $b_i$ (cf.~\eqref{eq:assum1}) i.e.,  $w(\p{c})  = b_i w_i(\p{c}) ~\forall \p{c} \in S_i$.
\STATE Return $\widehat{C}$, the set of $r$-subspaces that is an $\alpha$-approximate solution to the $(r, k)$-subspace clustering on input $(S, w)$.
\end{algorithmic}
\caption{Straggler-resilient $(r, k)$-subspace clustering}
\label{fig:subspaceClustering}
\end{algorithm}

Before we analyze the quality of $\widehat{C}$, the solution returned by Algorithm~\ref{fig:subspaceClustering}, we present the following result that shows the utility of an assignment scheme with Property~\ref{prop:assignment} to construct a global coreset for the entire dataset from the coresets of the partial datasets in a straggler-resilient manner.

\begin{lemma}\label{lem:weightedcoreset}
Let $P \subset \mathbb{R}^d$ be distributed according to $A$ with Property~\ref{prop:assignment}. Let $\p{b} \in \R^{|\calr|}$ be the recovery vector for the set of non-straggler nodes $\calr \subset [s]$. 
For any $i \in \calr$, let $S_i$ be an $\eps$-coreset for the local dataset $P_i$ with weight function $w_i: P_i \rightarrow \R$ with respect to the cost function ${\rm cost}(\cdot, \cdot)$. 
Then, $S:= \cup_{i \in \calr} S_i$ with the weight function $w: S \rightarrow \R$ defined as $w(\p{c}) = w_i(\p{c})\cdot b_i$ for all $\p{c} \in S_i$ is a $2(\eps+\delta)$-coreset for $P$. 
\end{lemma}

Since each $S_i$ is a $\delta$-coreset for $S_i$, it follows from Lemma~\ref{lem:weightedcoreset} that $S$ is a $4\delta$-coreset for $P$. This allows us to quantify the quality of $\widehat{C}$ as a solution to the underlying problem of $(r, k)$-subspace clustering on $P$.
\begin{theorem}\label{thm:coresetapprox}
Let ${\rm OPT}$ be the cost of the optimal $(r, k)$-subspace clustering solution for $P$. Then, we have that $\cost(P, \widehat{C}) \le \alpha (1+8 \delta) \cdot {\rm OPT}$.
\end{theorem}
\begin{proof}
Let $C^*$ be the optimal $(r, k)$-subspace clustering solution for $P$, i.e., $\cost(P, C^*) = {\rm OPT}$.  Since $S$ is a $4\delta$-coreset, we have
\begin{align*}
\cost(P, C) &\overset{(i)}{\le} \frac{ \cost(S, C, w)}{(1-4\delta)}
\overset{(ii)}{\le}  \frac{\alpha}{(1-4\delta)} \cost(S, C^*, w)\\
&\overset{(iii)}{\le}  \alpha \frac{(1+4\delta)}{(1-4\delta)} \cost(P, C^*) 
\le \alpha (1+8\delta) \text{OPT},
\end{align*}
where $(i)$ and $(iii)$ follow from Definition~\ref{def:coreset}; and $(ii)$ follows from the fact that the coordinator performs an $\alpha$-approximate subspace clustering on $(S, w)$.
\end{proof}
Coreset constructions for various clustering algorithms with squared $\ell_2$ error was considered by \cite{har2004coresets,har2004no, edwards2005no, frahling2005coresets, har2007smaller, chen2009coresets, langberg2010universal, feldman2011unified, varadarajan2012near}. However, the size of such constructed coresets depends on the ambient dimension $d$ which makes them prohibitive for high-dimensional datasets. Interestingly, Feldman et al. \cite{feldman13} and Balcan et al. \cite{balcan14} show that one could construct smaller sized relaxed coresets (of size independent of both $n$ and $d$) for the $r$-PCA problem. Further, they use these relaxed coresets for approximate-PCA to solve the class of $(r, k)$-subspace clustering problems. Below we show that, by utilizing a redundant assignment scheme with Property~\ref{prop:assignment}, the distributed approximate PCA and hence, $(r,k)$-subspace clustering algorithms of \cite{feldman13, balcan14} can be made straggler resilient.

\subsubsection{Straggler-resilient PCA using relaxed coresets}
\label{subsubsec:pca}
In this section, we use $P$ to denote both the set of $n$ data points and the $n \times d$ matrix with $n$ points in $P$ as its rows. We use the set and matrix notation interchangeably through this section.

The goal in PCA is to find the linear $r$-dimensional subspace $L$, that \emph{best fits} the data. It is well-known that the subspace spanned by the top $r$ right singular vectors of $P$ gives an optimal solution to the PCA problem. The main question addressed by Feldman et al. \cite{feldman13} is to find an approximate solution to PCA in a distributed setting by constructing relaxed coresets for the local data. In Algorithm~\ref{alg:pca}, we adapt the distributed PCA algorithm of \cite{feldman13} to obtain a straggler-resilient distributed PCA algorithm.
\vspace{0.4em}
\begin{algorithm}[h!]
\begin{algorithmic}[1]
\STATE Input: An $n \times d$ matrix $P$
\STATE Allocate $P$ to $s$ workers according $A$ with Property~\ref{prop:assignment}.
\STATE  Let $P_i$ denote the sub-matrix of $P$ contained at node $i$. Compute the SVD $P_i = U_i \Sigma_i V_i^T$.
\STATE For $r_1 = r + r/\delta -1$, define $\Sigma_i^{(r_1)}$ to be the matrix that contains first $r_1$ diagonal entries of $\Sigma_i$, and $0$ othewise. 
\STATE Send $S_i = \Sigma_i^{(r_1)} V_i^T$. 
\STATE Stack the $S_i$'s received from non-stragglers to 
construct $Y = [ S_i ]_{i \in \calr}$ and define $w(\p{y}) = b_i$, for $\p{y} \in S_i$.
\STATE Let $S = U \Sigma V^T$. 
\STATE Return $\widehat{L} = r\text{-PCA}(S, w)$. 
\end{algorithmic}
\caption{Straggler-resilient distributed $r$-PCA}
\label{alg:pca}
\end{algorithm}

The following result from \cite{feldman13} shows that each $S_i$ in Algorithm~\ref{alg:pca} is a relaxed $\delta$-coreset of $P_i$. Note the additive $\Delta_i$ term which is not present in the original definition of a coreset (cf.~Definition~\ref{def:coreset}).
\begin{lemma}[\cite{feldman13}]
\label{lem:feldmancoreset}
For any $i \in [s]$, let $ \widehat{P}_i$ denote the rows of  $U_i \Sigma_i^{(r_1)} V_i^T$. Then, for any $r$-dimensional linear subspace $L \subset \R^d$, there exists $\Delta_i = \Delta_i(P_i,  \widehat{P}_i)  > 0$ such that
\begin{align*}
\cost(P_i, L) \le \cost( \widehat{P}_i, L) + \Delta_i \le (1+\delta)\cost(P_i, L).
\end{align*}
\end{lemma}
Note that, for each $i \in [s]$, only the first $r_1$ rows of $S_i$ in Algorithm~\ref{alg:pca} are non-zero; as a result, each local relaxed coreset consists of $r_1$ vectors. Next, we show an equivalent of Lemma~\ref{lem:weightedcoreset} for relaxed coresets.
\begin{lemma}\label{lem:weightedRcoreset}
Let $\Delta:= \sum_{i \in \calr} b_i \Delta_i$. Then, for any $r$-dimensional linear subspace $L \subset \R^d$, 
\[
\cost(P, L) \le \cost(Y, L, w) + \Delta \le (1+4\delta) \cost(P, L).
\]
\end{lemma}
Now, Lemma~\ref{lem:weightedRcoreset} along with Lemma~\ref{lem:feldmancoreset} enable us to guarantee that the solution obtained from Algorithm~\ref{alg:pca} is close to the optimal for the distributed PCA problem.
\begin{theorem}\label{thm:pca}
Let $L^*$ be the optimal solution to $r$-PCA on $P$. Then, ${\rm cost}(P, \widehat{L}) \le (1+4\delta) {\rm cost}(P, L^*)$.
\end{theorem}

\begin{proof}
Note that $\Delta$ is independent of the choice of $L$. Thus, 
\begin{align*}
\cost(P, \widehat{L}) &\overset{(i)}{\le} \cost(Y, \widehat{L}, w) + \Delta \\
&\overset{(ii)}{\le} \cost(Y, L^*, w) + \Delta 
\overset{(i)}{\le} (1+ 4\delta) \cost(P, L^*),
\end{align*}
where $(i)$ and $(iii)$ follow from Lemma~\ref{lem:weightedRcoreset}; and $(ii)$ follows as $\widehat{L}$ is the optimal solution to the $r$-PCA problem on $Y$. 
\end{proof}


\subsection{Construction of assignment matrix}
\label{subsec:construction}
Finally, we present a randomized construction of the assignment matrix that satisfies Property~\ref{prop:assignment}. For the construction, we assume a random straggler model, where every compute node behaves as a straggler independently with probability $p_t$. Therefore, we receive the local computation from each compute node with probability $1-p_t$. 

Consider the following random ensemble of assignment matrices such that for some $\ell$ (to be chosen later) the $(i, j)$-th entry of the assignment matrix is defined as
\begin{align}
\label{eq:B_random}
A_{i,j} = \begin{cases}
1 & \text{with probability}~p_a = \frac{\ell}{s}, \\
0 & \text{with probability}~1 - p_a.
\end{cases}
\end{align} 

We show that for an appropriate choice of $\ell$, and hence $p_a$, the random matrix $A$ satisfies Property~\ref{prop:assignment} with high probability. 

\begin{theorem}\label{thm:randommatrix}
For any $\delta > 0$, the randomized assignment matrix (cf.~\eqref{eq:B_random}) with $\ell = \frac{6(2 + \delta)^2}{\delta^2}\cdot\frac{\log (\sqrt{2}n)}{1 - p_t}$ satisfies Property~\ref{prop:assignment} with probability at least $1 - \frac{1}{n}$ under the random straggler model.
\end{theorem}

The two parameters of importance when constructing an assignment matrix are the \emph{load per machine} and the \emph{fraction of stragglers} that can be tolerated. 
Increasing the redundancy makes the assignment matrix robust to more stragglers while at the same time, increases the computational load on individual compute nodes. 
For $s=O(n)$, our construction assigns $O(\log n)$ data points to each compute node and is resilient to a constant fraction of random stragglers. 

\begin{figure}[t!]
    \centering
\subfigure[Ground-truth.]{
\includegraphics[scale=0.3]{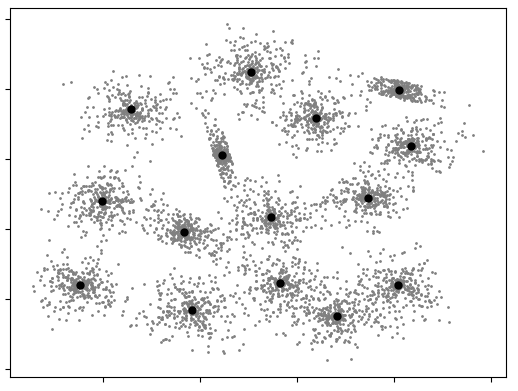}
\label{fig1}	
}%
\subfigure[No redundancy.]{
\includegraphics[scale=0.3]{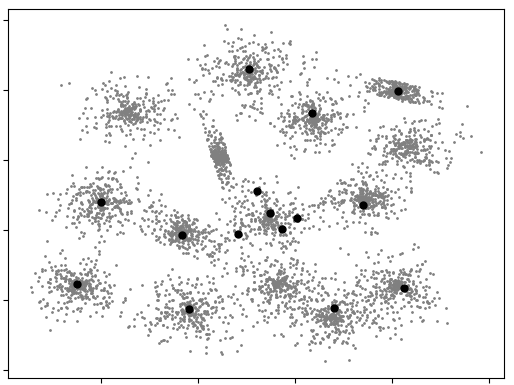}
\label{fig2}	
}%
\subfigure[$p_a = 0.1$.]{
\includegraphics[scale=0.3]{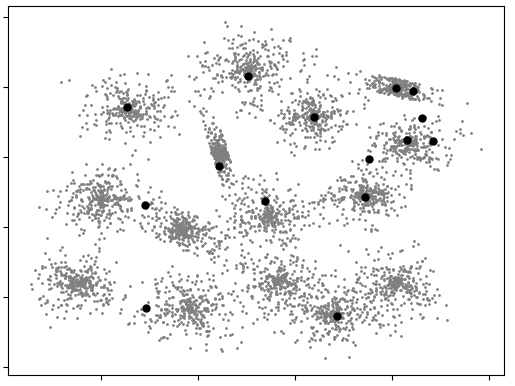}
\label{fig3}	
}%
\subfigure[$p_a = 0.2$.]{
\includegraphics[scale=0.3]{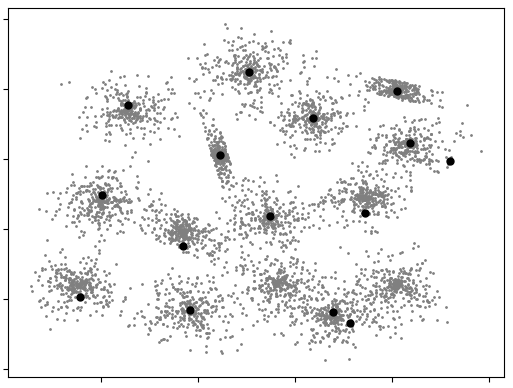}
\label{fig4}	
}%
\caption{Performance of the proposed straggler-resilient $k$-medians algorithm.}
\vspace{1em}
\end{figure}


\section{Experiments}
\label{sec:experiments}
In this section, we demonstrate the performance of our straggler-resilient distributed $k$-medians algorithm  and compare it with non-redundant data assignment scheme. We consider the synthetic Gaussian data-set \cite{Ssets} with $n=5000$ two-dimensional points. The points are distributed on $s=10$ compute nodes, $t=3$ of which are randomly chosen to be stragglers. The results are presented in Figures \ref{fig1}-\ref{fig4}. 

Figures \ref{fig1} shows the ground truth $k=15$-median clustering, using the centroids provided in the data-set. Figures \ref{fig2} shows the results obtained by ignoring the local computations from the straggler nodes. We use Algorithm~\ref{fig:distributedClustering} without any redundant data assignment. The $5000$ data points are randomly partitioned among $10$ compute nodes. Each non-straggler compute node sends its local $k$-median centers to the coordinator. The coordinator then runs a $k$-median algorithm on the accumulated $k(s-t)$ centers obtained from the non-stragglers. As evident from the comparison between Figure~\ref{fig1} and Figure~\ref{fig2}, such a scheme can output a set of poor quality $k$-centers. 

Figure \ref{fig3} shows the result of Algorithm~\ref{fig:distributedClustering}. The assignment matrix is chosen randomly (see Section~\ref{subsec:construction} for details) with each $p_a:= Pr[A_{i,j} = 1] = 0.1$. Such an assignment matrix assigns $500$ data points to each compute node in expectation, leading to a non-redundant data assignment. Figure \ref{fig4} shows the effect of increasing $p_a$ to $0.2$, and hence the redundancy in the data assignment step. Each compute node now gets about $1000$ data points. Note that the results of Figure \ref{fig4} are very close to the ground truth clustering shown in Figure \ref{fig1}.


\section{Conclusion and Future directions}
It is an interesting direction to explore the tradeoff between the communication and the approximation factor achieved by the clustering method. In the $k$-median algorithm described above, each compute node returns a set of $k$-centers to achieve an approximation factor of about $3$. Whereas, in Algorithm~\ref{fig:subspaceClustering}, the compute nodes send the coresets of their local data which can then be combined to construct a global coreset for the entire data. While, the quality of the obtained centers improves, it increases the communication cost between the compute nodes and the coordinator since a coreset would contain more than $k$ points. 

Another natural question that we are currently exploring is using data distribution techniques  to design distributed algorithms that are robust to byzantine adversaries. 
 

\newpage
\bibliographystyle{unsrt}
\bibliography{Coded_computation}


\newpage
\clearpage
\appendix

\section{Missing proofs from Section~\ref{sec:background}}

\begin{proof}[Proof of Theorem~\ref{thm:coreset-clustering}]
Let $C^*$ and $C^*_S$ be the optimal sets of $k$-centers for the clustering problem with the dataset $P$ and the weighted dataset $(S, w)$ as inputs, respectively.  Let $C$ be an $\alpha$-approximate solution for the clustering problem with the weighted dataset $(S, w)$. 
Our goal is to show that 
\begin{align}
\label{eq:coreset-clustering-1}
\cost(P, C) \le \alpha(1+3\eps) \cdot \cost(P, C^*).
\end{align}
Note that
\begin{align}
\label{eq:coreset-clustering-2}
\cost(P, C^*) &\overset{(i)}{\ge} \frac{1}{1+\eps} \cdot \cost(S, C^*, w) \nonumber \\
&\overset{(ii)}{\ge} \frac{1}{1+\eps} \cdot \cost(S, C^*_S, w) \nonumber \\
&\overset{(iii)}{\ge} \frac{1}{1+\eps} \cdot \frac{1}{\alpha} \cdot \cost(S, C, w)  \nonumber \\
&\overset{(iv)}{\ge} \frac{1}{1+\eps} \cdot \frac{1}{\alpha} \cdot (1-\eps) \cdot \cost(P,C),
\end{align}
where $(i)$ and $(iv)$ follow from the fact that $(S, w)$ is an $\epsilon$-coreset for $P$ (cf.~Defintion~\ref{def:coreset}). $(ii)$ and $(iii)$ follow from the fact that  $C^*_S$ and $C$ correspond to the optimal and $\alpha$-approximate clustering for $(S, w)$, respectively. Now, for small enough $\epsilon$ (in particular, $0 \leq \epsilon \leq {1}/{3}$), (\ref{eq:coreset-clustering-1}) follows from (\ref{eq:coreset-clustering-2}).
\end{proof}


\section{Missing proofs from Section~\ref{subsec:assignment}}

\begin{proof}[Proof of Lemma~\ref{lem:weightedcoreset}]
We prove the result for the $d^2(\cdot, \cdot)$ cost function, and the proof extends similarly to $d(\cdot, \cdot)$ as well. The proof is independent of the choice of the distance function, and we only use the properties of the assignment matrix. 

First note that, 
\begin{align}\label{eq:coresetweights}
\sum_{i \in \calr} b_i \text{cost} (P_i, C, w) &= \sum_{i \in \calr} b_i \sum_{\p{p} \in P_i} w(\p{p}) d^2(\p{p}, C) \nonumber\\
&= \sum_{i \in \calr} b_i \sum_{j \in [n]} A_{i, j} w(\p{p}_j)~d^2(\p{p}_j, C) \nonumber\\
&= \sum_{j \in [n]} w(\p{p}_j)~d^2(\p{p}_j, C)  \sum_{i \in \calr} b_i A_{i, j}. 
\end{align}
From Property~\ref{prop:assignment} we know that  for any $j \in [n]$, $\sum_{i \in \calr} b_i A_{i, j} \leq 1+\delta$. By combining this fact with \eqref{eq:coresetweights}, we obtain that
\begin{align*}
\sum_{i \in \calr} b_i \text{cost} (P_i, C, w)& = \sum_{j \in [n]} w(\p{p}_j)~d^2(\p{p}_j, C)  \sum_{i \in \calr} b_i A_{i, j}  \\
& \le (1+\delta) \sum_{j \in [n]} w(\p{p}_j)~d^2(\p{p}_j, C)\\
&= (1+\delta)\cdot \text{cost}(P, C, w)
\end{align*}
Similarly, Property~\ref{prop:assignment} ensures that for any $j \in [n]$, $\sum_{i \in \calr} b_i A_{i, j} \geq 1$. Utilizing this fact in \eqref{eq:coresetweights} gives us the desired lower bound as follows.
\begin{align*}
\sum_{i \in \calr} b_i \text{cost} (P_i, C, w) &= \sum_{j \in [n]} w(\p{p}_j)~d^2(\p{p}_j, C)  \sum_{i \in \calr} b_i A_{i, j}  \\
&\ge \sum_{j \in [n]}w(\p{p}_j)~d^2(\p{p}_j, C)\\
&=\text{cost}(P, C).
\end{align*}
\end{proof}

\section{Missing proofs from Section~\ref{subsec:kmeds}}
\begin{proof}[Proof of Lemma~\ref{lem:kmedianscoreset}]
~\\
\noindent\textbf{Upper Bound.}~We first show that for any set of $k$-centers $C \subset \R^d$, and any $j \in [s]$, 
$\text{cost}(Y_j,  C, w_j) \le 2 \text{cost}(P_j, C)$. This ensures that the weighted $k$-centers $(Y_j, w_j)$ is a good representation of the partial data $P_j$. 
\begin{align}\label{eq:kmedub1}
\text{cost}(Y_j, C, w_j) &= \sum_{\p{y} \in Y_j} w_j(\p{y}) d(\p{y}, C) \nonumber\\
&= \sum_{\p{y} \in Y_j} \lvert \text{cluster}(\p{y}, P_j) \rvert ~d(\p{y}, C) \nonumber\\
&= \sum_{\p{y} \in Y_j} \sum_{\p{x} \in \text{cluster}(\p{y}, P_j)} d(\p{y}, C).  
\end{align}
For any $\p{x} \in \R^d$, let $C(\p{x})$ denote its closest center in $C$. It follows from \eqref{eq:kmedub1} that
\begin{align}
\label{eq:kmedub2}
\text{cost}&(Y_j, C, w_j) = \sum_{\p{y} \in Y_j} \sum_{\p{x} \in \text{cluster}(\p{y}, P_j)} d(\p{y}, C(\p{y})) \nonumber\\
&\overset{(i)}{\le}  \sum_{\p{y} \in Y_j} \sum_{\p{x} \in \text{cluster}(\p{y}, P_j)} d(\p{y}, C(\p{x})) \nonumber\\
&\overset{(ii)}{\le} \sum_{\p{y} \in Y_j} \sum_{\p{x} \in \text{cluster}(\p{y}, P_j)} \left( d(\p{x},\p{y})+d(\p{x}, C(\p{x}))\right) \nonumber\\ 
&=  \sum_{\p{y} \in Y_j} \sum_{\p{x} \in \text{cluster}(\p{y}, P_j)} d(\p{x},\p{y})+ \sum_{\p{x} \in P_j}  d(\p{x}, C(\p{x}))\nonumber \\
&= \text{cost}(P_j, Y_j) + \text{cost}(P_j, C) \nonumber\\
&\overset{(iii)}{\le} 2 ~\text{cost}(P_j, C),
\end{align}
where $(i)$ and $(ii)$ employ the definition of $C(\p{x})$ and the triangle inequality, respectively. Note that $(iii)$ follows from the optimality of the $k$-centers $Y_j$ on the partial  dataset $P_j$, i.e., $\cost(P_j, Y_j) \leq \cost(P_j, C)$. Next, note that
\begin{align}
\label{eq:lemma_iii_up}
&\cost(Y, C, w) = \sum_{j\in \calr} \text{cost}(Y_j, C, b_j\cdot w_j) \nonumber \\
&\qquad = \sum_{j\in \calr} b_j \text{cost}(Y_j, C, w_j) \nonumber \\
&\qquad \overset{(i)}{\le} 2 \sum_{j\in \calr} b_j ~\text{cost}(P_j, C) \overset{(ii)}{\le} 2(1+\delta)\text{cost}(P, C),
\end{align}
where $(i)$ and $(ii)$ follow from \eqref{eq:kmedub2} and Lemma~\ref{lem:weightedcoreset}, respectively.\\

\noindent \textbf{Lower Bound.}~To establish the lower bound, we start from Lemma~\ref{lem:weightedcoreset}. For any set of $k$-centers $C$, we have
\begin{align}\label{eq:kmedlb1}
\text{cost}(P, C) &\le \sum_{j \in \calr} b_j \text{cost} (P_j, C)\nonumber\\
&= \sum_{j \in \calr} b_j \sum_{\p{x} \in P_j} d(\p{x}, C(\p{x}))
\end{align}
Recall that $Y_j$ is the set of $k$-median centers for the data-set $P_j$. By the definition of cluster centers, we know that for any two points $\p{x}, \p{y} \in \R^d$, and any set of $k$-centers, $d(\p{x}, C(\p{x})) \le d(\p{x}, C(\p{y}))$. Plugging this observation in \eqref{eq:kmedlb1}, we obtain that
\begin{align}\label{eq:kmedlb2}
\text{cost}(P, C) &\le \sum_{j \in \calr} b_j \sum_{\p{x} \in P_j} d(\p{x}, C(\p{x})) \nonumber\\
&\le \sum_{j \in \calr} b_j \sum_{\p{x} \in P_j} d\big(\p{x}, C( Y_j (\p{x}))\big), 
\end{align}
where $Y_j(\p{x})$ is the cluster center in $Y_j$ that is closest to $\p{x} \in P_j$. Now using triangle inequality, we get that 
\begin{align}
\label{eq:lemma_iii_lower}
\text{cost}(P, C) &\le \sum_{j \in \calr} b_j \sum_{\p{x} \in P_j} d\big(\p{x}, C( Y_j (\p{x}))\big) \nonumber \\
&\overset{(i)}{\le} \sum_{j \in \calr} b_j \sum_{\p{x} \in P_j} \Big( d(\p{x}, Y_j (\p{x})) +  d\big(Y_j (\p{x}), C( Y_j (\p{x}))\big)  \Big) \nonumber  \\
&= \sum_{j \in \calr} b_j \text{cost}(P_j, Y_j) + \sum_{j \in \calr} b_j \sum_{\p{x} \in P_j}d(Y_j (\p{x}), C( Y_j (\p{x}))) \nonumber \\
&= \sum_{j \in \calr} b_j \text{cost}(P_j, Y_j) + \sum_{j \in \calr} b_j \sum_{\p{y} \in Y_j}\lvert \text{cluster}(\p{y}, P_j) \rvert d(\p{y}, C(\p{y})) \nonumber \\
&= \sum_{j \in \calr} b_j \text{cost}(P_j, Y_j)  + \sum_{j \in \calr} b_j \text{cost}(Y_j, C, w_j)\nonumber \\
&= \sum_{j \in \calr} b_j \text{cost}(P_j, Y_j)  + \sum_{j \in \calr} \text{cost}(Y_j, C, b_j\cdot w_j)\nonumber \\
&= \sum_{j \in \calr} b_j \text{cost}(P_j, Y_j)  +  \text{cost}(Y, C, w),
\end{align}
where $(i)$ employs the triangle inequality.

Note that Lemma~\ref{lem:kmedianscoreset} follows from \eqref{eq:lemma_iii_up} and \eqref{eq:lemma_iii_lower}. 
\end{proof}

\section*{Missing Proofs from Section~\ref{subsec:pca}}
\begin{proof}[Proof of Lemma~\ref{lem:weightedcoreset}]
Note that, for any $i \in \calr$, the weighted point set $(S_i, w_i)$ is an $\eps$-coreset of the partial dataset $P_i$. Thus, according to Definition~\ref{def:coreset}, we have
\begin{align}\label{eq:coresetpartial}
(1-\eps)~{\rm cost}(P_i, C) \le ~{\rm cost}&(S_i, C, w_i)   \le (1+\eps)~{\rm cost}(P_i, C),
\end{align}
for any set of $k$-centers $C \subset \R^d$. 

For $S:= \cup_{i \in \calr} S_i$ and any set of $k$-centers $C$, we have
\begin{align}
\label{eq:thm_weighted_1}
\cost(S, C, w) &= \sum_{\p{c} \in S} w(\p{c}) d^2(\p{c}, C) \nonumber \\
&= \sum_{i \in \calr} b_i \sum_{\p{c} \in S_i} w_i(\p{c}) d^2(\p{c}, C)   \nonumber \\
&= \sum_{i \in \calr} b_i \cost(S_i, C, w_i). 
\end{align} 
By combining \eqref{eq:coresetpartial} and \eqref{eq:thm_weighted_1}, we obtain that
\begin{align*}
(1-\eps) \sum_{i \in \calr} b_i~{\rm cost}(P_i, C) 
& \le  \cost(S, C, w) \le (1+\eps) \sum_{i \in \calr} b_i ~{\rm cost}(P_i, C) \\
\end{align*} 
Now using Lemma~\ref{lem:weightedcoreset} on both sides of the above inequality, we obtain that
\begin{align}
\label{eq:thm_weighted_2}
\cost(S, C, w) &\ge (1-\eps) \sum_{i \in \calr} b_i~{\rm cost}(P_i, C) \nonumber \\
& \ge (1-\eps)\cost(P, C) \nonumber \\
&  \ge (1 - 2 \eps - 2\delta)\cost(P, C),
\end{align}
and
\begin{align}
\label{eq:thm_weighted_3}
\cost(S, C, w) &\le (1+\eps) \sum_{i \in \calr} b_i~{\rm cost}(P_i, C) \nonumber \\
&\le (1+\eps)(1+\delta)\cost(P, C) \nonumber \\
&\le (1+ 2 \eps + 2\delta)\cost(P, C).
\end{align}
Lemma~\ref{lem:weightedcoreset} follows from \eqref{eq:thm_weighted_2} and \eqref{eq:thm_weighted_3}.
\end{proof}

\begin{proof}[Proof of Lemma~\ref{lem:weightedRcoreset}]

From Algorithm~\ref{alg:pca}, note that
\begin{align*}
\cost(S, L, w) &= \sum_{j \in \calr} b_j \cost(S_j, L)\\
&=\sum_{j \in \calr} b_j \cost(\widehat{P}_j, L).
\end{align*}
The last equality follows from the observation that $\widehat{P}_j = U_j S_j$, where $U_j$ is an orthonormal matrix. Therefore, $\text{cost}(\widehat{P}_j, L) = \text{cost}(S_j, L)$ for any $r$-dimensional subspace $L$.

By invoking Lemma~\ref{lem:feldmancoreset}, we obtain that
\begin{align}
\label{eq:lemm_weighted_1}
\cost(S, L, w)&=\sum_{j \in \calr} b_j \cost(\widehat{P}_j, L) \nonumber \\
&\le \sum_{j \in \calr} b_j \left( (1+\delta)\cost(P_j, L) - \Delta_j \right) \nonumber \\
&= (1+\delta) \sum_{j \in \calr} b_j \cost(P_j, L) - \Delta \nonumber \\
&\overset{(i)}{\le} (1+\delta)^2 \cost(P, L) - \Delta \nonumber \\
& \le (1+ 4\delta) \cost(P, L) - \Delta,
\end{align}
where $(i)$ follows from Lemma~\ref{lem:weightedcoreset}. Similarly, again using Lemma~\ref{lem:weightedcoreset}, we obtain that
\begin{align}
\label{eq:lemm_weighted_2}
\cost(S, L, w)&=\sum_{j \in \calr} b_j \cost(\widehat{P}_j, L) \nonumber \\
&\ge \cost(P, L) - \Delta.
\end{align}
Note that Lemma~\ref{lem:weightedRcoreset} follows from \eqref{eq:lemm_weighted_1} and \eqref{eq:lemm_weighted_2}.
\end{proof}

\section*{Missing Proofs from Section~\ref{subsec:construction}}

\begin{proof}[Proof of Theorem~\ref{thm:randommatrix}]
Let $\calr \subseteq [s]$ denote the set of non-stragglers, then for any $i \in [s]$, we have 
\begin{align}
\label{eq:w_random}
\Pr\{i \in \calr \} = 1 - p_t.
\end{align}
Next, we argue that for any $\delta > 0$, we can choose $p_a = \frac{\ell}{s}$ large enough to ensure Property~\ref{prop:assignment} with high probability. First, we analyze the weight of each of the column in the random matrix $A_{\calr}$. For $i \in [s]$ and $j \in [n]$, define an event $E_{i,j}$ as follows
\begin{align}
E_{i,j} = \begin{cases}
1 & \text{if $i \in \calr$ and $A_{i,j} = 1$,}\\
0 & \text{otherwise}.
\end{cases}
\end{align}
Note that for any fixed $j \in [n]$, $\big\{E_{i,j}\big\}_{i \in [s]}$ is a collection of $s$ independent events. Furthermore, it follows from \eqref{eq:B_random} and \eqref{eq:w_random} that 
\begin{align*}
\Pr\{E_{i,j} = 1\} = p_a(1 - p_t).
\end{align*}
Note that 
$$
\mathbb{E}\left[\sum_{i = 1}^{s}E_{i,j}\right] = s~p_{a}(1 - p_t) = {\ell}(1 - p_t).
$$
It then follows from standard Chernoff bound that for any $\gamma \in (0, 1)$,
\begin{align*}
\Pr&\Big\{\bigl\vert\sum_{i = 1}^{s}E_{i,j}  - \ell(1 - p_t)\bigr\vert \geq \gamma \cdot\ell(1 - p_t)\Big\} \nonumber\\
& \leq 2e^{- \frac{\gamma^2\cdot\ell(1 - p_t)}{3}}.
\end{align*}

In particular, if we choose $\gamma = \frac{\delta}{2  + \delta}$ and $\ell = \frac{6\log (\sqrt{2}n)}{\gamma^2\cdot(1 - p_t)}$, then with probability at least $1 - \frac{1}{n^2}$ the following holds for a given $j \in [n]$.
\begin{align*}
1 \leq \frac{1}{(1 - \gamma)\cdot\ell  (1 - p_t)} \sum_{i = 1}^{s}E_{i,j} \leq 1 + \delta.
\end{align*}
Now, taking a union bound over all $j \in [n]$, we have with probability at least $1 - \frac{1}{n}$,
\begin{align}
\label{eq:whp_event1}
1 \leq \frac{1}{(1 - \gamma)\cdot\ell  (1 - p_t)} \sum_{i = 1}^{s}E_{i,j} \leq 1 + \delta~~\forall~j \in [n].
\end{align}

Recall that in order to establish Property~\ref{prop:assignment}, we need to show that there exist a non-negative vector $\p{b} \in \R^{|\calr|}$ such that 
$$
\p{b}^TA_{\calr} = (a_1, a_2,\ldots, a_n), 
$$
where
$$
1 \leq a_j \leq 1 + \delta~~\text{for all}~j \in [n].
$$
We consider the vector 
$$
\p{b} = \frac{1}{(1- \gamma)\cdot\ell(1-p_t)}\cdot(1, 1,\ldots, 1)
$$
as a candidate. Note that for this choice of $\p{b}$, we have 
$$
\p{b}^TB_{\calr} = \frac{1}{(1 - \gamma)\cdot\ell(1- p_t)}\cdot\big(\sum_{i = 1}^{s}E_{i,1}, \ldots, \sum_{i = 1}^{s}E_{i,n}\big).
$$
It follows from \eqref{eq:whp_event1}, that with probability at least $1 - \frac{1}{n}$, each of the coordinates of $\p{b}^{T}B_{\calr}$ falls in the interval $[1, 1 + \delta]$. This completes the proof.
\end{proof}

\end{document}